\documentclass[sn-mathphys-num]{sn-jnl}

\usepackage{graphicx}%
\usepackage{multirow}%
\usepackage{amsmath,amssymb,amsfonts}%
\usepackage{amsthm}%
\usepackage{mathrsfs}%
\usepackage[title]{appendix}%
\usepackage{xcolor}%
\usepackage{textcomp}%
\usepackage{manyfoot}%
\usepackage{booktabs}%

\usepackage{listings}%

\usepackage{amsmath,amsfonts,amsbsy,amssymb,amsthm,mathtools,tabulary,graphicx,caption,fancyhdr, mathrsfs, pgf, pgfplots, epsfig, epstopdf, lscape,booktabs,enumitem,lipsum,tkz-graph,scalerel,stackengine,setspace,adjustbox,float,amsthm}
\usepackage[labelformat=simple]{subfig}

\theoremstyle{thmstyleone}%
\newtheorem{theorem}{Theorem}
\theoremstyle{thmstyletwo}%

\theoremstyle{thmstylethree}%
\newtheorem{definition}{Definition}%

\begin{document}

\title[Generalized Derangetropy Functionals for Modeling Cyclical Information Flow]{Generalized Derangetropy Functionals for Modeling Cyclical Information Flow}


\author*[1]{\fnm{Masoud} \sur{Ataei}}\email{masoud.ataei@utoronto.ca}

\author[2]{\fnm{Xiaogang} \sur{Wang}}\email{stevenw@yorku.ca}

\affil[1]{\normalsize\orgdiv{Department of Mathematical and Computational Sciences},\\ \orgname{University of Toronto}, \state{Ontario}, \country{Canada}}

\affil[2]{\normalsize\orgdiv{Department of Mathematics and Statistics}, \orgname{York University}, \state{Ontario}, \country{Canada}}



\abstract{This paper introduces a  framework for modeling cyclical and feedback-driven information flow through a generalized family of entropy-modulated transformations called derangetropy functionals. Unlike scalar and static entropy measures such as Shannon entropy, these functionals act directly on probability densities and provide a topographical representation of information structure across the support of the distribution. The framework captures periodic and self-referential aspects of information distribution and encodes them through functional operators governed by nonlinear differential equations. When applied recursively, these operators induce a spectral diffusion process governed by the heat equation, leading to convergence toward a Gaussian characteristic function. This convergence theorem provides a unified analytical foundation for describing the long-term dynamics of information under cyclic modulation. The proposed framework offers new tools for analyzing the temporal evolution of information in systems characterized by periodic structure, stochastic feedback, and delayed interaction, with applications in artificial neural networks, communication theory, and non-equilibrium statistical mechanics.
 }




\maketitle

\section{Introduction}	
The propagation of information within complex systems often manifests through intricate temporal and structural patterns, marked by recurrence, feedback, and modulated fluctuations. Such behavior is especially prevalent in biological networks, cognitive architectures, and cyclic physical processes, where probability distributions evolve under local stochastic rules, coupled with global periodic or geometrical constraints. While classical entropy measures, particularly Shannon entropy \cite{Shannon1948}, serve as foundational tools in quantifying uncertainty, their scalar and static nature renders them inadequate for capturing the dynamic evolution of informational states in feedback-driven or periodically modulated systems. This shortcoming is particularly pronounced in domains such as artificial neural networks or biological circuits, where probability densities evolve through temporally structured interactions.

In response to these limitations, the present work introduces a generalized framework for derangetropy, a class of entropy-modulated functional transformations designed to model cyclical and feedback-induced information flow. Extending prior developments in derangetropy theory \cite{ataei2024derangetropy}, we delineate and analyze three distinct types of derangetropy functionals, each encoding a specific mode of probabilistic modulation. Type-I derangetropy induces an entropy-attenuating transformation that sharpens low-entropy regions while preserving the underlying support of the original distribution. This type is well-suited to systems exhibiting structural persistence under dynamical accumulation, such as resting-state neurophysiological activity. In contrast, Type-II derangetropy acts as an entropy-amplifying transformation, accentuating high-entropy regions, thus capturing phenomena governed by stochastic dispersion or instability, as in turbulent flows or volatile economic systems.

Type-III derangetropy departs from entropy coupling and instead introduces a phase-modulated transformation, redistributing probability mass via a quadratic sine function. This structure is particularly effective in modeling resonance-driven dynamics, as observed in electrophysiological systems or harmonic oscillators. A notable property of derangetropy is its asymptotic behavior: under recursive application, it induces a diffusion process governed by the heat equation, with convergence to a Gaussian characteristic function. This convergence theorem not only unifies the three derangetropy classes but also furnishes a robust analytical tool for studying information dynamics in systems subject to periodic forcing or structured feedback.

The present framework situates derangetropy within a broader effort to characterize information flow in systems governed by cyclical structure or delayed interactions. Recent developments in this area span multiple disciplines. In neuroscience, oscillatory cortical dynamics are recognized as critical for synchronization and information routing, necessitating analytical tools capable of capturing the redistribution of informational content over time \cite{buzsaki2004neuronal}. In information theory, communication over channels with cyclostationary noise requires capacity analyses that incorporate periodic variability \cite{shlezinger2019capacity, dabora2023capacity}. Similarly, in control systems and machine learning, recurrent feedback architectures give rise to non-Markovian dynamics, prompting the use of directed information and other causal metrics \cite{kramer1998directed, quinn2011estimating}, although these methods often lack a global functional formulation capturing the evolution of distributional support.

Within statistical physics and non-equilibrium thermodynamics, mutual information has been employed to quantify exchange between coupled subsystems \cite{vin2009mutual}, and bipartite models have been advanced to link informational exchange with entropy production in open systems \cite{horowitz2014thermodynamics}. These ideas have also been extended to quantum domains, where coherence and entanglement alter the informational landscape \cite{ptaszynski2019thermodynamics}. Adaptive models for information flow in ecologies and engineered systems likewise reflect sensitivity to environmental variation \cite{nicoletti2021mutual}, underscoring the need for tools that encode temporal plasticity and structural feedback. The information dynamics in financial systems have been further investigated through the lens of transfer entropy, where it is established that a bidirectional causal relationship exists between the underlying processes governing realized and implied volatilities of the stock market \cite{ataei2021theory}.

The generalized derangetropy functionals developed herein contribute to this body of work by enabling localized redistribution of probability mass through functionals that incorporate feedback, periodicity, and entropy coupling. The universal convergence theorem, central to this framework, provides an analytical basis for interpreting long-term behavior under recursive application, thereby connecting local transformations with global dynamical outcomes. By integrating structural and dynamical features into a unified model, this approach extends traditional entropy measures and affords novel insight into the evolution and stabilization of information in complex environments.

The paper proceeds as follows. In Section~\ref{Sec:Type-I}, we define the Type-I derangetropy and derive its associated governing differential equation and self-referential properties. Section~\ref{Sec:Type-II} addresses the Type-II derangetropy, characterizing its entropy-amplifying transformation and the associated nonlinear flow dynamics. Section~\ref{Sec:Type-III} develops the Type-III derangetropy and presents a convergence analysis emphasizing Gaussianization and spectral diffusion. We conclude in Section~\ref{Sec:Conclusion} with a discussion of prospective applications in structured probabilistic models and extensions to interacting dynamical systems.

\section{Type-I Derangetropy}
\label{Sec:Type-I}
\subsection{Definition}

Consider a probability space $(\Omega, \mathscr{F}, \mathbb{P})$, where $\Omega$ denotes the sample space, $\mathscr{F}$ is a $\sigma$-algebra of measurable subsets of $\Omega$, and $\mathbb{P}$ is a probability measure defined on $\mathscr{F}$. Let $X: \Omega \rightarrow \mathbb{R}$ be a real-valued random variable that is measurable with respect to the Borel $\sigma$-algebra $\mathscr{B}_{\mathbb{R}}$ on $\mathbb{R}$. Suppose that $X$ has an absolutely continuous distribution with an associated probability density function (PDF) $f \in$ $\mathcal{L}^2\left(\mathbb{R}, \mathscr{B}_{\mathbb{R}}, \lambda\right)$, where $\lambda$ represents the Lebesgue measure on $\mathbb{R}$. The cumulative distribution function (CDF) corresponding to $X$ is given by
\begin{equation*}
	F(x) = \int_{-\infty}^x f(t) \, d\lambda(t), \quad x \in \mathbb{R}.
\end{equation*}

We now introduce the notion of Type-I derangetropy, which refines probability distributions while incorporating entropy-based modulation.
\begin{definition}[Type-I Derangetropy]
	\label{Def:TypeI}
	The \textit{Type-I derangetropy functional} $\rho: \mathcal{L}^2(\mathbb{R}, \mathscr{B}_{\mathbb{R}}, \lambda) \rightarrow \mathcal{L}^2(\mathbb{R}, \mathscr{B}_{\mathbb{R}}, \lambda)$ is defined by the following mapping
	\begin{equation}
		\rho[f](x) = \left(\frac{24}{\pi e}\right) \sin(\pi F(x)) F(x)^{F(x)} (1-F(x))^{1-F(x)}  f(x).
	\end{equation}
	Alternatively, it can be rewritten as a Fourier-type transformation
	\begin{equation}
		\rho[f](x) = \left(\frac{24}{\pi e}\right) \sin(\pi F(x))\,  e^{-H_B(F(x))} f(x),
	\end{equation}
	where
	\begin{equation}
		H_B(F(x)) = -F(x) \log (F(x)) - (1-F(x)) \log (1-F(x)),
	\end{equation}
	is the \textit{Shannon entropy} for a Bernoulli distribution with success probability $p = F(x)$. The~evaluation of the derangetropy functional of Type-I at a specific point $x \in \mathbb{R}$ is denoted by $\rho_f(x)$.
\end{definition}

The sine function in the transformation plays a crucial role in systems that exhibit periodicity, where information alternates between concentrated and dispersed states. This function serves as both a modulation mechanism and a projection operator. As a modulation mechanism, it introduces controlled oscillations that reflect inherent periodic characteristics in the probability distribution. The oscillatory nature of the sine function ensures that the transformation redistributes probability mass while preserving essential features of the original density. As a projection operator, it maps probability densities into a transformed space where local variations become more prominent, allowing for a more refined probabilistic representation. This behavior is particularly relevant in cyclic systems such as time-series processes and wave-based phenomena, where periodic effects shape the behavior of probability distributions.

The presence of $H_B(F(x))$ in the transformation introduces a localized uncertainty measure, quantifying how probability mass is distributed at any given point. This entropy function describes the informational balance between the cumulative probability $F(x)$ and its complement $1-F(x)$, providing a measure of relative uncertainty in the distribution. By incorporating entropy directly into the transformation, Type-I derangetropy adapts the probability density in a manner influenced by local uncertainty. The weighting factor $e^{-H_B(F(x))}$ modulates the extent to which the transformation affects the distribution, attenuating its influence in high-entropy regions while allowing stronger reshaping in low-entropy areas. This adaptive adjustment ensures that the transformation refines the distribution without enforcing uniform smoothing, preserving key probabilistic characteristics.

The explicit presence of $f(x)$ in the transformation ensures that Type-I derangetropy refines probability densities rather than reconstructing them entirely. Unlike convolution-based smoothing techniques or regularization approaches, this transformation maintains the local structural integrity of the distribution while selectively modifying probability mass based on entropy considerations. In regions of low entropy, the transformation retains the original density, preventing unnecessary redistribution. In contrast, in high-entropy regions, where uncertainty is greater, it modulates probability values in a controlled manner, refining the density without excessive distortion. This property distinguishes Type-I derangetropy from conventional entropy-based transformations that often impose global constraints or flatten distributions.

Type-I derangetropy is best suited for systems that exhibit a combination of periodicity and self-regulation, where information accumulates in a structured yet dynamic fashion. A canonical example lies in biological neural systems, such as resting-state EEG or fMRI signals, where regions of the brain alternate between high and low informational states. In such systems, the entropy-attenuating nature of Type-I ensures that stable, low-entropy zones retain their shape while more uncertain, transitional regions are selectively reshaped. This allows for the extraction of latent functional hierarchies that are otherwise obscured by global entropy measures like Shannon or R\'enyi entropy. In practice, when applied to empirical distributions over space or time, Type-I derangetropy acts as a local contrast enhancer, refining meaningful variations without flattening them, ideal for denoising or revealing subtle periodic signatures in oscillatory-yet-organized domains, such as circadian gene expressions, ecological cycles, or cortical oscillations.

\subsection{Mathematical Properties}
We now establish key mathematical properties of the Type-I derangetropy functional. For any absolutely continuous PDF $f(x)$, the functional $\rho[f](x)$ is a nonlinear operator that belongs to the space $C^{\infty}(\mathbb{R})$ and remains a valid probability density function.
\begin{theorem}
	For any absolutely continuous $f(x)$, the~derangetropy functional $\rho_f(x)$ is a valid probability density function.
\end{theorem}

\begin{proof}
	To prove that $\rho_f(x)$ is a valid PDF, we need to show that $\rho_f(x) \geq 0$ for all $x \in \mathbb{R}$ and that $\int_{-\infty}^{\infty} \rho_f(x) d x=1$. The non-negativity of $\rho_f(x)$ is clear due to the non-negativity of its components. To verify the normalization condition, a change of variable $z=F(x)$ is applied. Since $d z=f(x) d x$, this transformation allows the integral of $\rho_f(x)$ to be rewritten in terms of $z$ as follows:
	\begin{equation}
		\int_{-\infty}^{\infty} \rho_f(x) d x=\int_0^1\left(\frac{24}{\pi e}\right) \sin (\pi z) z^z(1-z)^{1-z} d z.
	\end{equation}
	As established in \cite{ataei2024derangetropy}, the integral
	\begin{equation}
		\int_0^1 \sin (\pi z) z^z(1-z)^{1-z} d z=\frac{\pi e}{24},
	\end{equation}	
	which, in turn, implies that	
	\begin{equation}
		\int_{-\infty}^{\infty} \rho_f(x) d x=1.
	\end{equation}	
	Hence, $\rho_f(x)$ is a valid PDF.
\end{proof}

Since $\rho_f(x)$ is a valid PDF, it encapsulates the informational content of the original distribution $f(x)$ while recursively uncovering a hierarchical structure of its own information content. This iterative refinement is governed by the following recurrence relation
\begin{equation}
	\rho_f^{(n)}(x) = \left(\frac{24}{\pi e}\right) \sin\left(\pi \mathcal{G}_f^{(n-1)}(x)\right) e^{-H_B(\mathcal{G}_f^{(n-1)}(x))} \rho_f^{(n-1)}(x),
\end{equation}
where $\rho_f^{(n)}(x)$ represents the $n$th iteration of the derangetropy functional, and
$$
\mathcal{G}_f^{(n)}(x)=\int_{-\infty}^x \rho_f^{(n)}(t) d t
$$
is the associated CDF. The initial conditions for this recursion are given by
$$
\rho_f^{(0)}(x)=f(x) \quad \text{and} \quad \mathcal{G}_f^{(0)}(x)=F(x),
$$
where $f(x)$ and $F(x)$ denote the PDF and CDF of the original distribution, respectively. This recursive structure ensures that each subsequent iteration $\rho_f^{(n)}(x)$ integrates the informational content of all preceding layers, thereby constructing a hierarchical probabilistic representation of information.

Moreover, the self-referential hierarchical structure distinguishes Type-I derangetropy from conventional entropy-based transformations, which typically impose global constraints or flatten distributions. Furthermore, since the transformation retains absolute continuity, the process maintains smoothness across iterations, allowing for a deeper probabilistic analysis over successive refinements.

The Type-I derangetropy functional exhibits a structured refinement process, systematically redistributing probability mass while preserving fundamental probabilistic properties. This refinement is characterized by a second-order nonlinear differential equation that governs the evolution of probability accumulation in distribution functions. The interaction between entropy-based modulation and probabilistic redistribution ensures that the transformation dynamically adjusts the density function while retaining key informational characteristics.

An illustrative case emerges when considering a uniform distribution over the interval $(0,1)$, where the PDF is given by $f(x)=1$ and the corresponding CDF is $F(x)=x$. Under these conditions, the governing differential equation encapsulates an intricate equilibrium between entropy suppression, logistic accumulation, and oscillatory refinement. 

\begin{theorem}
	Let $X$ be a random variable following a uniform distribution on the interval (0, 1). Then, the Type-I derangetropy functional $\rho_f(x)$ satisfies the following {\color{black}second-order nonlinear ordinary differential equation}
	\begin{equation} \frac{d^2 \rho_f(x)}{d F(x)^2} + 2 \log \left(\frac{1 - F(x)}{F(x)}\right) \frac{d \rho_f(x)}{d F(x)} + \left[\pi^2 - \frac{1}{F(x)(1 - F(x))} + \log^2 \left(\frac{1 - F(x)}{F(x)}\right) \right] \rho_f(x) = 0, \end{equation}
	where the initial conditions are set as
	$$
	\left.\rho_f(x)\right|_{F(x)=0}=0, \quad \text { and }\left.\quad \frac{d \rho_f(x)}{d F(x)}\right|_{F(x)=0}=\frac{24}{e}.
	$$
\end{theorem}
\begin{proof}
	To solve this differential equation, we introduce an integrating factor $\mu(F(x))$ to eliminate the first-order derivative term. As established in \cite{ataei2024derangetropy}, this integrating factor is given by	
	\begin{equation}
		\mu(F(x))=e^{F(x)\log(F(x)) + (1-F(x))\log(1-F(x))}.
	\end{equation}	
	By rewriting the exponential term, this expression simplifies to
	\begin{equation}
		\mu(F(x))=F(x)^{F(x)}(1-F(x))^{1-F(x)} .
	\end{equation}	
	Multiplying the entire equation by this integrating factor transforms it into a simpler form, allowing us to define a new function $v(F(x))$ such that	
	\begin{equation}
		\rho_f(x)=\mu(F(x)) v(F(x)).
	\end{equation}	
	Substituting this transformation into the differential equation, the first-order derivative term cancels, reducing the equation to	
	\begin{equation}
		\frac{d^2 v(F(x))}{d F(x)^2}+\pi^2 v(F(x))=0 .
	\end{equation}
	This equation is a standard second-order homogeneous linear differential equation with constant coefficients. Its general solution is	
	\begin{equation}
		v(F(x))=C_1 \sin (\pi F(x))+C_2 \cos (\pi F(x)),
	\end{equation}	
	where $C_1$ and $C_2$ are constants. Finally, evaluating the initial values yields $C_1=\frac{24}{\pi e}$ and $C_2=0$, implying that $\rho_f(x)$ is indeed a solution to the differential equation.
\end{proof}

The differential equation governing Type-I derangetropy describes how probability density evolves under its transformation, revealing a deep connection between probability accumulation and entropy modulation. The presence of the log-odds function, $\log \left(\frac{1-F(x)}{F(x)}\right)$, naturally arises in entropy-driven transformations and provides insight into the underlying refinement process. Since this function encodes the relative probability mass distribution, its role in the differential equation suggests that the rate of change in probability density is governed by the balance of mass accumulation within the total probability space. The transformation redistributes probability mass in a way that maintains equilibrium between entropy modulation and structured probabilistic refinement, ensuring that high-uncertainty regions are dynamically adjusted while preserving the essential probabilistic features of the original distribution.

The derivative of the log-odds function, given by
\begin{equation}
	\frac{d}{d F(x)} \log \left(\frac{1-F(x)}{F(x)}\right)=\frac{d}{d F(x)} H_B(F(x))=\frac{1}{F(x)(1-F(x))}
\end{equation}
quantifies the local entropy gradient, which dictates the strength of the refinement process. The transformation is most active near $F(x)=0.5$, where entropy gradients reach their maximum, and progressively weaker as $F(x) \rightarrow 0$ or $F(x) \rightarrow 1$, where entropy naturally diminishes. This implies that the transformation does not merely smooth probability distributions but rather introduces structured refinements that reflect localized entropy variations. In other words, the transformation modulates probability mass according to the entropy landscape, amplifying changes where uncertainty is highest and stabilizing regions where probability mass is more concentrated. The function $\log \left(\frac{1-F(x)}{F(x)}\right)$ is anti-symmetric about $F(x)=0.5$, indicating that the governing equation differentiates probability mass accumulation on either side of the median. This property ensures that refinements occur in a balanced manner, preventing distortions in probability redistribution. That is, the structured modulation introduced by Type-I derangetropy ensures that probability adjustments remain coherent across the entire distribution, avoiding artificial drifts that would compromise the probabilistic integrity of the refined density function.

The presence of the term $\pi^2$ in the differential equation highlights the intrinsic oscillatory nature of the refinement process. This term ensures that the transformation does not merely act as a smoothing operator but instead introduces a structured modulation mechanism that preserves critical features of the original distribution. The oscillatory behavior dictated by this term prevents excessive entropy flattening while allowing the transformation to refine probability densities dynamically. As a result, the governing equation encodes an evolving probability transformation where entropy regulation, probabilistic redistribution, and oscillatory modulation contribute to a structured refinement process that maintains coherence in probability mass adjustments.

\section{Type-II Derangetropy}
\label{Sec:Type-II}
\subsection{Definition}
The Type-II derangetropy functional represents a refinement process that enhances uncertainty while maintaining the overall probabilistic structure of the distribution. Unlike its Type-I counterpart, which redistributes probability mass in a manner that mitigates uncertainty, Type-II derangetropy operates through an entropy-amplifying transformation that reinforces the contribution of high-entropy regions. The notion of the Type-II derangetropy functional is formally defined below.

\begin{definition}[Type-II Derangetropy]
	\label{Def:TypeII}
	The \textit{Type-II derangetropy functional} $\tau: \mathcal{L}^2(\mathbb{R}, \mathscr{B}_{\mathbb{R}}, \lambda) \rightarrow \mathcal{L}^2(\mathbb{R}, \mathscr{B}_{\mathbb{R}}, \lambda)$ is defined by the following mapping
	\begin{equation}
		\tau[f](x) = \left(\frac{e}{\pi}\right) \dfrac{\sin(\pi F(x))}{F(x)^{F(x)} (1-F(x))^{1-F(x)}}   f(x),
	\end{equation}
	or, alternatively, as a Fourier-type transformation
	\begin{equation}
		\tau[f](x) = \left(\frac{e}{\pi}\right) \sin(\pi F(x))\,  e^{H_B(F(x))} f(x).
	\end{equation}
	The~evaluation of the derangetropy functional of Type-II at a specific point $x \in \mathbb{R}$ is denoted by $\tau_f(x)$.
\end{definition}

The defining characteristic of Type-II derangetropy is its ability to amplify entropy through the weighting factor $e^{H_B(F(x))}$, which enhances probability mass in high-entropy regions while reducing its relative influence in areas of low entropy. Unlike Type-I derangetropy, which suppresses entropy to structure probability redistribution, Type-II derangetropy increases entropy by emphasizing uncertainty-driven refinements. Rather than redistributing probability mass uniformly, the transformation reshapes the density in a way that increases the contribution of regions where entropy is higher while reducing the dominance of more concentrated areas.

The sinusoidal term $\sin (\pi F(x))$ ensures that the transformation remains structured by regulating how entropy amplification modifies probability mass. Without this modulation, the entropy-amplified density $e^{H_B(F(x))} f(x)$ would introduce large variations, potentially distorting the overall density. The sine function interacts with entropy amplification to control the redistribution process, reinforcing probability density adjustments while preventing excessive deviation from the original distribution. This interaction ensures that refinements follow a controlled transformation rather than an unbounded expansion of probability mass.

Unlike Type-I derangetropy, which refines probability density while constraining uncertainty, Type-II derangetropy amplifies entropy while directing the redistribution of probability mass according to its entropy profile. This does not merely shift probability mass to high-entropy regions but increases their relative contribution while reducing the influence of low-entropy areas. The transformation does not lead to indiscriminate dispersion but instead enforces an entropy-driven redistribution that aligns with the intrinsic uncertainty of the distribution.

The self-regulating nature of Type-II derangetropy ensures that entropy amplification follows an adaptive process where probability mass is redistributed in accordance with entropy gradients. Instead of increasing entropy arbitrarily, the transformation adjusts the probability density function in a way that preserves its structural properties while refining its information content. Unlike standard transformations that impose smoothing effects or regularization constraints, Type-II derangetropy enhances entropy while preserving meaningful distinctions within the probability distribution. This controlled entropy amplification suggests applications in probability density refinement, entropy-aware transformations, and statistical modeling, where increasing entropy while maintaining structural coherence is crucial for probabilistic inference.

Type-II derangetropy becomes relevant in contexts where high-entropy regions are not noise to be suppressed but rather signals of inherent system dynamics, such as in stochastic processes, turbulent flows, or volatile financial markets. By amplifying entropy instead of attenuating it, Type-II captures the natural tendency of these systems to drift into disordered or diffusive states. For instance, in the modeling of high-frequency trading behavior or biophysical systems under random forcing (e.g., ion channel kinetics, intracellular transport), local entropy peaks are often indicative of regime transitions or instability fronts. Type-II derangetropy allows these zones to be magnified, helping to identify precursors of phase change or collapse. Additionally, its application is beneficial in epidemiological modeling of contagion, where regions of maximal uncertainty (e.g., transmission rate variability) are critical to the propagation mechanism and thus should be weighted more heavily in information-theoretic assessments. In essence, Type-II is a tool for systems where informational disorder is not just present, but mechanistically central.

\subsection{Empirical Observations}
The behavior of Type-I and Type-II derangetropy functionals highlights fundamental differences in probability mass refinement. Each functional encodes a distinct form of entropy modulation, shaping probability density according to structured probability evolution and entropy-aware transformations. Figure \ref{fig:Derangetropy_Distributions} examines five representative distributions-uniform, normal, exponential, semicircle, and arcsine-illustrating how these functionals respond to symmetry, skewness, and boundary effects.

For the uniform distribution, where probability is evenly distributed across its support, Type-I derangetropy reshapes it into a semi-parabolic density, peaking at the median where entropy is most balanced. As probability mass moves toward the boundaries, entropy decreases, leading to a decline in functional values. The Type-II transformation further enhances central probability concentration, resulting in a bell-shaped density that decays more sharply near the edges.

For the normal distribution, both transformations preserve symmetry while refining probability mass around the mean. Type-I derangetropy applies a smooth redistribution, maintaining the general shape, while Type-II amplifies the central peak, increasing probability concentration in high-entropy regions and compressing lower-probability areas.

For the exponential distribution, which exhibits strong skewness, both functionals adjust the density while retaining its asymmetry. Type-I derangetropy smooths out the distribution, moderating the steep decline at low values while preserving the overall shape. Type-II intensifies probability concentration in the high-density region, further emphasizing entropy-driven refinements while reducing influence in lower-probability areas.

For the semicircle distribution, both functionals preserve its compact support and symmetry but differ in their refinement strategies. Type-I maintains the overall shape while introducing mild adjustments near the median. Type-II amplifies probability mass at the center, reducing density near the boundaries, reflecting its tendency to reinforce entropy-dominant regions more aggressively.

For the arcsine distribution, where probability mass is highly concentrated at the boundaries, the transformations yield particularly distinct refinements. Type-I derangetropy redistributes mass away from the edges, producing an inverted semi-parabolic density that spreads probability more evenly. Type-II generates a semi-parabolic shape peaking at the median, significantly reducing mass at the boundaries and reinforcing entropy-maximizing regions. This contrast underscores the functional distinction: while Type-I ensures proportional redistribution, Type-II concentrates mass in high-entropy areas, diminishing lower-entropy contributions.

These observations demonstrate that Type-I derangetropy functions as an entropy-balancing transformation, refining probability distributions while preserving structural proportions. Type-II derangetropy, in contrast, acts as an entropy-enhancing transformation, intensifying probability mass concentration in high-entropy regions while suppressing low-entropy contributions. This distinction makes Type-II particularly suited for applications requiring stronger probability centralization, while Type-I is more appropriate for cases demanding balanced entropy refinement.

\begin{figure}
	
	\subfloat{{\includegraphics[scale=0.55]{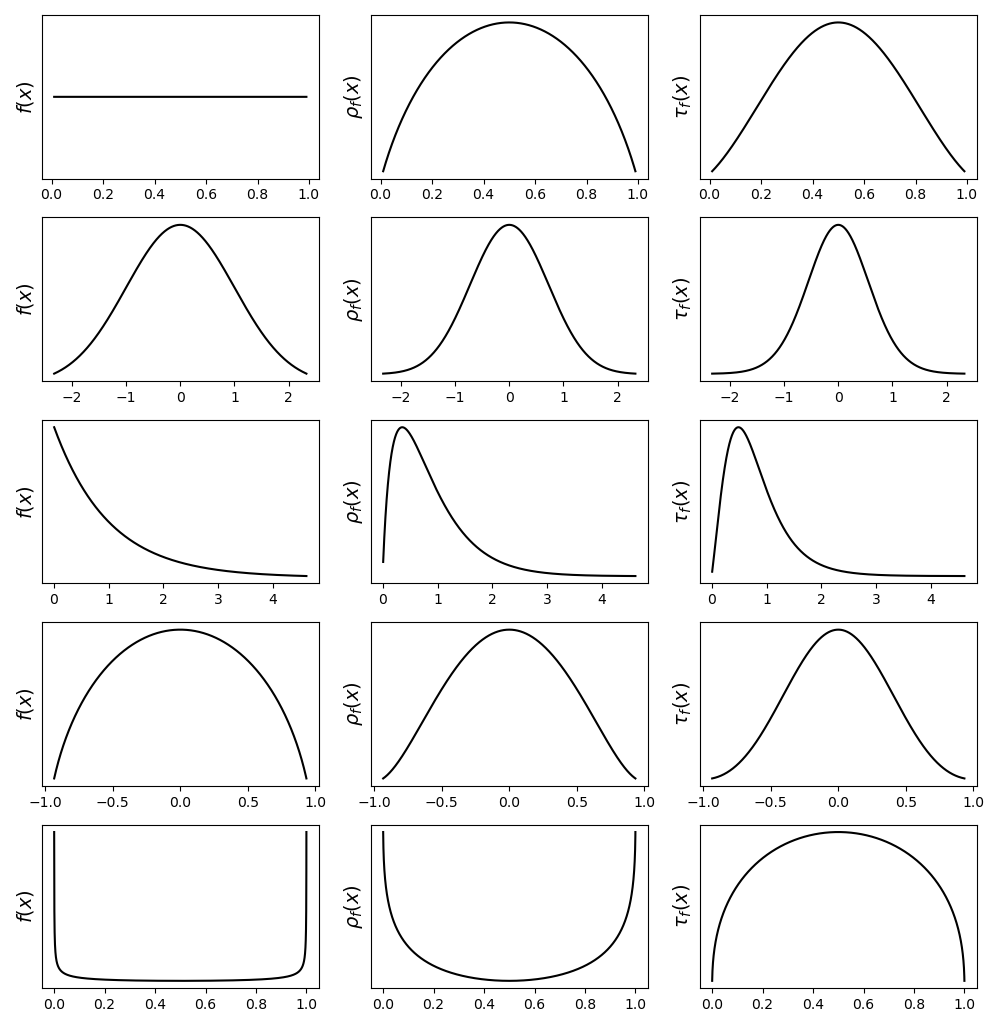} }}%
	
	\caption{{Plots of}
		probability density functions $f(x)$ (\textbf{left}), Type-I derangetropy functionals $\rho_f(x)$ (\textbf{middle}), and Type-II derangetropy functionals $\tau_f(x)$ (\textbf{right}) for $\operatorname{uniform}\, (0,1)$, $\operatorname{normal}\, (0,1)$, $\operatorname{exponential} \,(1)$, $\operatorname{semicircle} \,(-1,1)$, and~$\operatorname{arcsine}\, (0,1)$ distributions.}%
	\label{fig:Derangetropy_Distributions}%
\end{figure}

\subsection{Mathematical Properties}

We now establish key mathematical properties of the Type-II derangetropy functional that underscore its utility in analyzing probability distributions. For any absolutely continuous PDF $f(x)$, the derangetropy functional $\tau[f](x)$ is a nonlinear operator that belongs to the space $C^{\infty}(\mathbb{R})$ having the following first derivative\vspace{6pt}
\begin{equation}
	\label{Eq:Derivative}
	\frac{d}{dx} \tau_f(x) = \tau_f(x) \left[ \pi \cot(\pi F(x)) + \log\left(\frac{F(x)}{1-F(x)}\right) + \frac{f'(x)}{f(x)} \right],
\end{equation}
where the derivatives are taken with respect to $x$. 

The~following theorem shows that the Type-II derangetropy maps the PDF $f(x)$ into another valid probability density function.
\begin{theorem}
	For any absolutely continuous $f(x)$, the Type-II derangetropy functional $\tau_f(x)$ is a valid probability density function.
\end{theorem}
\begin{proof}
	To prove that $\tau_f(x)$ is a valid PDF, we need to show that $\tau_f(x) \geq 0$ for all $x \in \mathbb{R}$ and that $\int_{-\infty}^{\infty} \tau_f(x) \, dx = 1$. The~non-negativity of $\tau_f(x)$ is clear due to the non-negativity of the terms involved in its definition. The~normalization condition can further be verified by the change of variables $z = F(x)$, yielding
	\begin{equation}
		\int_{-\infty}^{\infty} \tau_f(x) \, dx = \int_{0}^{1} \left(\frac{e}{\pi}\right) \dfrac{\sin (\pi z)}{z^z(1-z)^{1-z}}  dz.
	\end{equation}	
	{As shown} 
	in Appendix A, the~integral
	\begin{equation}
		\int_{0}^{1} \dfrac{\sin (\pi z)}{z^z(1-z)^{1-z}} dz = \frac{\pi}{e},
	\end{equation}	
	which, in~turn, implies that
	\begin{equation}
		\int_{-\infty}^{\infty} \tau_f(x) \, dx = 1.
	\end{equation}
	Hence, $\tau_f(x)$ is a valid PDF.
\end{proof}

As a valid probability density function, the Type-II derangetropy functional $\tau_f(x)$ enables the formulation of a self-referential framework, where the informational content of each successive layer is recursively expressed in terms of previous layers. This hierarchical structure captures how probability mass evolves under repeated applications of Type-II derangetropy, encoding both entropy amplification and oscillatory refinement.

The recursive formulation follows from the transformation:
$$
\tau_f^{(n)}(x)=\left(\frac{e}{\pi}\right) \sin \left(\pi \mathcal{G}_f^{(n-1)}(x)\right) e^{H_B\left(\mathcal{G}_f^{(n-1)}(x)\right)} \tau_f^{(n-1)}(x)
$$
where $\tau_f^{(n)}(x)$ represents the $n$th iteration of the transformation, and
$$
\mathcal{G}_f^{(n)}(x)=\int_{-\infty}^x \tau_f^{(n)}(t) d t
$$
is the CDF at the $n$th stage. The recursion is initialized by setting
$$
\tau_f^{(0)}(x)=f(x) \quad \text{and} \quad \mathcal{G}_f^{(0)}(x)=F(x).
$$

This formulation ensures that each subsequent layer $\tau_f^{(n)}(x)$ encapsulates the cumulative effect of all preceding transformations. The iterative refinement process progressively modulates probability redistribution based on the evolving entropy landscape, reinforcing high-entropy regions while systematically adjusting lower-entropy areas. This recursive structure highlights the fundamental role of Type-II derangetropy as an adaptive entropy-enhancing transformation, allowing probability distributions to evolve through a controlled sequence of refinements.

The Type-II derangetropy functional further satisfies the following differential equation.
\begin{theorem}
	Let $X$ be a random variable following a uniform distribution on the interval (0, 1). Then, the~derangetropy functional $\tau_f(x)$ satisfies the following {\color{black}second-order nonlinear ordinary differential equation}
	\begin{equation} \frac{d^2 \tau_f(x)}{d F(x)^2} - 2 \log \left(\frac{1 - F(x)}{F(x)}\right) \frac{d \tau_f(x)}{d F(x)} + \left[\pi^2 + \frac{1}{F(x)(1 - F(x))} + \log^2 \left(\frac{1 - F(x)}{F(x)}\right) \right] \tau_f(x) = 0, \end{equation}
	where the initial conditions are set as
	$$
	\left.\rho_f(x)\right|_{F(x)=0}=0, \quad \text { and }\left.\quad \frac{d \rho_f(x)}{d F(x)}\right|_{F(x)=0}=e.
	$$
\end{theorem}
\begin{proof}
	To solve the differential equation, we introduce an integrating factor $\mu(F(x))$ to eliminate the first-order derivative term, which is given by
	\begin{equation}
		\mu(F(x))=e^{F(x) \log \left(\frac{1-F(x)}{F(x)}\right)}(1-F(x))^{-1}.
	\end{equation}
	Rewriting the exponential term, this simplifies to
	\begin{equation}
		\mu(F(x))=F(x)^{-F(x)}(1-F(x))^{-(1-F(x))}.
	\end{equation}	
	Multiplying the entire equation by this integrating factor transforms it into a simpler form, allowing us to define a new function $v(F(x))$ such that	
	$$
	\tau_f(x)=\mu(F(x)) v(F(x)).
	$$
	Substituting this transformation into the differential equation, the first-order derivative term cancels, reducing the equation to	
	$$
	\frac{d^2 v(F(x))}{d F(x)^2}+\pi^2 v(F(x))=0.
	$$
	This equation is a second-order homogeneous linear differential equation with constant coefficients, whose general solution is	
	\begin{equation}
		v(F(x))=C_1 \sin (\pi F(x))+C_2 \cos (\pi F(x)),
	\end{equation}
	where $C_1$ and $C_2$ are constants.
	To determine these constants, we apply the initial conditions, which yields $C_1=\frac{e}{\pi}$ and $C_2=0$, implying that $\tau_f(x)$ is indeed a solution to the differential equation.
\end{proof}

The governing differential equations for Type-I and Type-II derangetropy reveal a fundamental mathematical duality, illustrating how both transformations share the same underlying structure while diverging in their treatment of entropy gradients and probability evolution. This structural connection is encoded in the opposing signs of key terms, specifically, the first-order derivative coefficient and the entropy-dependent potential, which dictate distinct modes of probability refinement.

The log-odds function, appearing in both equations, serves as the fundamental driver of probability redistribution, quantifying the relative balance of probability mass on either side of a given point. In Type-I derangetropy, the positive coefficient of the first derivative term ensures that entropy gradients modulate probability refinement in a balanced manner, preventing excessive concentration while redistributing probability mass proportionally across the domain. The negative entropy correction term reinforces this structured refinement, ensuring that mass is neither excessively centralized nor overdispersed.

The differential equation for Type-II derangetropy introduces an inverse transformation by reversing the sign of the first derivative term, altering the interaction between probability mass and entropy gradients. This reversal amplifies the influence of high-entropy regions, strengthening probability concentration where uncertainty is greatest while reducing mass in low-entropy regions. The positive entropy correction term further reinforces this effect, ensuring that the transformation enhances contrast in probability distribution, rather than merely diffusing mass. Unlike Type-I derangetropy, which ensures entropy balancing across the domain, Type-II derangetropy enforces entropy amplification, strengthening dominant probability regions and reducing the presence of structured low-entropy areas.

This duality between the governing differential equations highlights the complementary nature of TypeI and Type-II derangetropy as entropically conjugate processes. Under Type-I derangetropy, the evolution of probability mass follows a structured refinement, ensuring entropy-controlled redistribution without excessive dominance of high-entropy regions. In contrast, under Type-II derangetropy, probability mass evolution reinforces entropy amplification, emphasizing refinements that concentrate probability mass in regions of highest uncertainty while suppressing low-entropy influence. The governing equations formalize this interaction between structured refinement and entropy-enhancing transformations, demonstrating that while distinct in their mechanisms, these transformations together define a unified framework for entropy-aware probability evolution.

\section{Type-III Derangetropy}
\label{Sec:Type-III}
\subsection{Definition}

The Type-III derangetropy functional introduces a novel probability transformation based on quadratic sine modulation of the probability density function.

\begin{definition}[Type-III Derangetropy]
	\label{Def:TypeIII}
	The \textit{Type-III derangetropy functional} $\nu: \mathcal{L}^2(\mathbb{R}, \mathscr{B}_{\mathbb{R}}, \lambda) \rightarrow \mathcal{L}^2(\mathbb{R}, \mathscr{B}_{\mathbb{R}}, \lambda)$ is defined by the following mapping
	\begin{equation}
		\nu[f](x) = 2 \sin^2(\pi F(x))\  f(x).
	\end{equation}
	The~evaluation of the derangetropy functional of Type-III at a specific point $x \in \mathbb{R}$ is denoted by $\nu_f(x)$.
\end{definition}

This transformation introduces a wave-like probability refinement mechanism, redistributing probability mass while maintaining a structured balance. The function $\sin ^2(\pi F(x))$ serves as a nonlinear modulation function, ensuring that probability mass is amplified near central regions while being suppressed at the boundaries.

Unlike entropy-based refinements, which redistribute probability mass in response to local uncertainty, Type-III derangetropy introduces a purely oscillatory modulation mechanism. It applies a sinusoidal weighting, implementing a structured amplification-suppression scheme: mid-range cumulative probabilities receive maximal weighting, while regions near the extremes are progressively attenuated. Notably, the transformation is bounded within the interval $[0,2]$, ensuring that it preserves global integrability and prevents unbounded growth across iterations. This boundedness gives rise to a self-limiting behavior, allowing the transformation to refine local features while maintaining the overall probabilistic structure.

In contrast to Type-I and Type-II derangetropies?which explicitly couple transformation intensity to local entropy?Type-III is entropy-agnostic. Its quadratic sine kernel enforces central amplification, redistributing probability mass in a phase-symmetric fashion that is sensitive to the waveform geometry of the cumulative distribution. This property renders Type-III particularly effective for tasks involving probability sharpening, contrast enhancement, and density regularization, especially where entropy modulation is undesirable or misleading.

Type-III derangetropy is well-suited for phase-dominated systems, where resonance and periodic structure govern information flow. These include applications in electrophysiological signal analysis, communication waveforms, and mechanical or optical oscillators, where preservation of coherent phase relationships is critical. Because the transformation operates independently of entropy, it enables neutral redistribution of density with respect to intrinsic oscillatory modes. This is especially valuable in neuroengineering, for isolating biomarkers such as alpha, beta, and gamma bands, or in seismic and speech signal processing, where it facilitates the detection of cyclic modulations without imposing artificial smoothness. Furthermore, under repeated application, the transformation converges toward a Gaussian profile in the Fourier domain, echoing the spectral diffusion observed in many dissipative dynamical systems. This convergence establishes Type-III derangetropy as a natural operator for spectral embedding, phase tracking, and harmonic extraction within complex cyclic environments.

\subsection{Mathematical Properties}

For~any absolutely continuous PDF $\nu(x)$, the Type-III derangetropy functional $\nu[f](x)$ is a nonlinear operator that belongs to the space $C^{\infty}(\mathbb{R})$ having the following first derivative\vspace{6pt}
\begin{equation}
	\label{Eq:Derivative}
	\frac{d}{dx} \nu_f(x) = \nu_f(x) \left[ 2\pi \cot(\pi F(x)) f(x) + \frac{f'(x)}{f(x)} \right],
\end{equation}
where the derivatives are taken with respect to $x$. 

The~following theorem shows that the Type-III derangetropy maps the PDF $f(x)$ into another valid probability density function.
\begin{theorem}
	For any absolutely continuous $f(x)$, the Type-III derangetropy functional $\nu_f(x)$ is a valid probability density function.
\end{theorem}
\begin{proof}
	To prove that $\nu_f(x)$ is a valid PDF, we need to show that $\nu_f(x) \geq 0$ for all $x \in \mathbb{R}$ and that $\int_{-\infty}^{\infty} \nu_f(x) \, dx = 1$. The~non-negativity of $\nu_f(x)$ is clear due to the non-negativity of the terms involved in its definition. The~normalization condition can further be verified by the change of variables $z = F(x)$, yielding
	\begin{equation}
		\int_{-\infty}^{\infty} \nu_f(x) \, dx = \int_{0}^{1} 2 \sin^2(\pi z)\   dz=\int_0^1(1-\cos (2 \pi z)) d z= 1.
	\end{equation}	
	Hence, $\nu_f(x)$ is a valid PDF.
\end{proof}

Similar to its Type-I and Type-II counterparts, the Type-III derangetropy functional $\nu_f(x)$ has self-referential properties, which are expressed as follows: 
$$
\nu_f^{(n)}(x)=2 \sin^2 \left(\pi \mathcal{G}_f^{(n-1)}(x)\right) \tau_f^{(n-1)}(x)
$$
where $\nu_f^{(n)}(x)$ represents the $n$th iteration of the transformation, and
$$
\mathcal{G}_f^{(n)}(x)=\int_{-\infty}^x \nu_f^{(n)}(t) d t
$$
is the CDF at the $n$th stage. The initial conditions for the recursion are set as
$$
\nu_f^{(0)}(x)=f(x) \quad \text{and} \quad \mathcal{G}_f^{(0)}(x)=F(x).
$$

The recursive formulation of Type-III derangetropy highlights a structured refinement process where each successive iteration enhances the probabilistic modulation applied in the previous step. Unlike Type-I, which balances entropy gradients, and Type-II, which amplifies entropy regions, Type-III enforces a periodic redistribution mechanism governed by the sine-squared function. This transformation progressively sharpens probability mass near mid-range values of the CDF while suppressing contributions from extreme regions. As a result, with each iteration, the density function becomes increasingly concentrated in regions where $\sin ^2(\pi F(x))$ is maximized, leading to a self-reinforcing oscillatory refinement.

Figure \ref{fig:Derangetropy_TypeIII} illustrates how successive iterations of Type-III derangetropy transform different probability distributions. Despite the absence of explicit entropy terms, the iterative application of Type-III derangetropy exhibits a clear pattern of redistribution that progressively refines probability mass. The transformation does not merely preserve the shape of the original density but instead pushes probability toward a more structured, centralized form. Regardless of the initial distribution, the repeated application of Type-III derangetropy leads to a progressively smoother and more bell-shaped structure, ultimately converging toward a normal distribution.

For the uniform distribution, the initial transformation reshapes the density into a symmetric bell-like function, emphasizing probability concentration near the median. The second iteration intensifies this effect, reducing mass at the edges and further reinforcing the central region. A similar pattern is observed in the normal distribution, where each iteration enhances probability mass concentration near the mean while reducing the density of the tails. The effect is more apparent in skewed distributions such as the exponential distribution, where the first transformation shifts probability mass away from the heavy tail, while the second iteration further reduces the density gradient, making the distribution increasingly symmetric.

This process is evident in compact distributions such as the semicircle, where each transformation reduces the impact of boundary effects, favoring a more centralized redistribution of probability mass. In the case of the arcsine distribution, which is initially highly concentrated at the boundaries, the transformation effectively mitigates this extreme localization, redistributing probability toward the interior and gradually shifting toward a bell-shaped distribution. The common trend across all distributions is that Type-III derangetropy iteratively smooths and reshapes probability densities in a structured manner, reinforcing central regions while gradually eliminating extremities. The transformation acts as a refinement mechanism that reinforces a well-defined probabilistic structure, ensuring that the density evolves in a direction consistent with the Gaussian limit.

\begin{figure}
	
	\subfloat{{\includegraphics[scale=0.55]{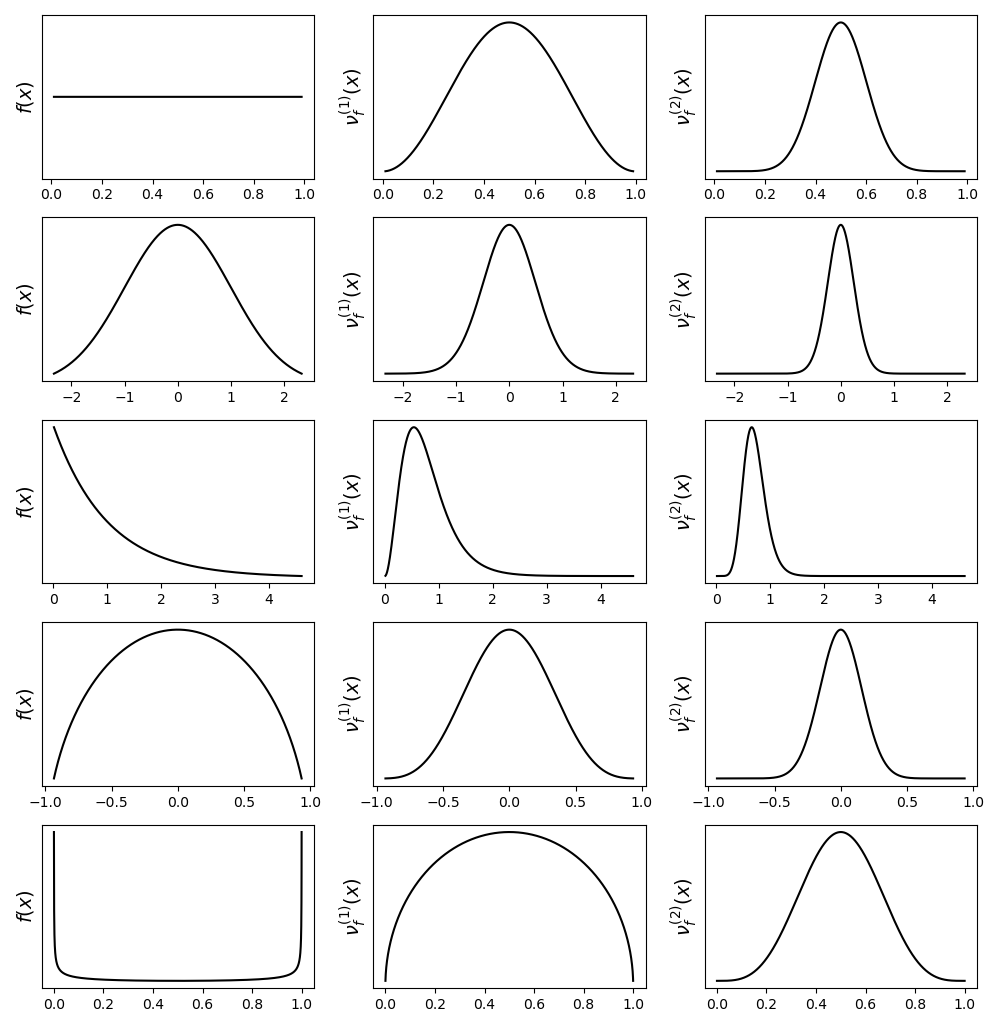} }}%
	
	\caption{{Plots of}
		probability density functions $f(x)$ (\textbf{left}), first-level Type-III derangetropy functionals $\nu_f^{(1)}(x)$ (\textbf{middle}), and second-level Type-III derangetropy functionals $\nu_f^{(2)}(x)$ (\textbf{right}) for $\operatorname{uniform}\, (0,1)$, $\operatorname{normal}\, (0,1)$, $\operatorname{exponential} \,(1)$, $\operatorname{semicircle} \,(-1,1)$, and~$\operatorname{arcsine}\, (0,1)$ distributions.}%
	\label{fig:Derangetropy_TypeIII}%
\end{figure}

The differential equation governing the Type-III derangetropy functional for a uniform distribution, as stated in the following theorem, reveals fundamental principles of wave-based probability refinement and structured probability evolution.

\begin{theorem}
	Let $X$ be a random variable following a uniform distribution on the interval $(0,1)$. Then, the~derangetropy functional $\nu_f(x)$ satisfies the following {\color{black}third-order ordinary differential equation}
	\begin{equation} \frac{d^3 \nu_f(x)}{d F(x)^3}+ 4\pi^2 \frac{d \nu_f(x)}{d F(x)} = 0, \end{equation}
	where the initial conditions are set as
	$$
	\left.\nu_f(x)\right|_{F(x)=0}=0, \quad \left.\quad \frac{d \nu_f(x)}{d F(x)}\right|_{F(x)=0}=0 \quad \text { and }\left.\quad \frac{d^2 \nu_f(x)}{d F(x)^2}\right|_{F(x)=0}=4\pi^2.
	$$
\end{theorem}
\begin{proof}
	The first and third derivatives of $\nu_f(x)$ are computed as follows:	
	\begin{equation}
		\frac{d \nu_f(x)}{d F(x)}=2 \pi \sin (2 \pi F(x)),
	\end{equation}
	and
	\begin{equation}
		\frac{d^3 \nu_f(x)}{d F(x)^3}=-8 \pi^3 \sin (2 \pi F(x)).
	\end{equation}
	Substituting these expressions into the differential equation,	we obtain
	\begin{equation}
		-8 \pi^3 \sin (2 \pi F(x))+4 \pi^2(2 \pi \sin (2 \pi F(x)))=0.
	\end{equation}
	Since both terms cancel, it follows that $\nu_f(x)$ satisfies the given third-order differential equation.
\end{proof}

The governing equation for Type-III derangetropy is a third-order linear differential equation whose characteristic equation,
\begin{equation}
	r^3+4 \pi^2 r=0,
\end{equation}
whose roots are
\begin{equation}
	r_1=0, \quad \text{and} \quad r_2, r_3= \pm 2 \pi i.
\end{equation}
This indicates a combination of stationary and oscillatory modes, implying that the probability refinement process follows a wave-like evolution rather than simple diffusion. Unlike purely diffusive models, which allow probability mass to spread indefinitely, the presence of oscillatory solutions implies a cyclic probability refinement mechanism, where mass is systematically redistributed in a structured manner.

The third-order nature of this equation introduces higher-order control over probability redistribution, capturing not only the rate of change of probability mass but also its jerk (the rate of change of acceleration). This distinguishes Type-III derangetropy from simple diffusion-based models, which typically involve only first- or second-order derivatives. The third derivative term ensures dynamic feedback regulation, where probability mass alternates between expansion and contraction, maintaining a structured oscillatory refinement process.

The term $4 \pi^2 \frac{d \nu_f}{d F}$ plays a crucial role in reinforcing oscillatory stabilization. Unlike a conventional damping mechanism, which suppresses oscillations over time, this term acts as a restorative force, ensuring that probability oscillations remain well-structured while preventing uncontrolled dispersion. This prevents the transformation from excessively smoothing out the density, instead enforcing a controlled wave-based refinement of probability mass. This structured evolution explains why Type-III derangetropy systematically enhances mid-range probability densities while suppressing boundary effects, leading to a self-regulating probability distribution that resists both excessive diffusion and over-concentration.

\subsection{Spectral Representation}

The Type-III derangetropy functional defines a structured probability transformation governed by a frequency-modulated mechanism. As a valid probability density function, this transformation induces a well-defined characteristic function, encapsulating its spectral properties. The following theorem establishes the explicit form of the characteristic function associated with Type-III derangetropy.

\begin{theorem}
	The characteristic function of the distribution induced by the Type-III derangetropy transformation is given by
	$$
	\varphi_{\nu}(t)=\varphi_0(t)-\frac{1}{2}\left[\varphi_F^{+}(t)+\varphi_F^{-}(t)\right],
	$$	
	where
	$$
	\varphi_0(t)=\int_{-\infty}^{\infty} e^{i t x} f(x) d x,
	$$	
	is the characteristic function of the original density, and	
	$$
	\varphi_F^{+}(t)=\int_{-\infty}^{\infty} e^{i t x} e^{i 2 \pi F(x)} f(x) d x, \quad \varphi_F^{-}(t)=\int_{-\infty}^{\infty} e^{i t x} e^{-i 2 \pi F(x)} f(x) d x,
	$$	
	are modulated characteristic functions incorporating phase shifts.
\end{theorem}

\begin{proof}
	See Appendix B.
\end{proof}

Applying this result to the special case of a uniform distribution, we obtain
\begin{equation}
	\varphi_\nu(t)=\frac{e^{i t}-1}{i t}-\frac{1}{2}\left(\frac{e^{i(t+2 \pi)}-1}{i(t+2 \pi)}+\frac{e^{i(t-2 \pi)}-1}{i(t-2 \pi)}\right).
\end{equation}
The characteristic function of the Type-III derangetropy functional for a uniform distribution follows a structured transformation process, which can be expressed in terms of the characteristic function of the base distribution. Given that the characteristic function of a uniform random variable is
\begin{equation}
	\varphi_0(t)=\frac{e^{i t}-1}{i t},
\end{equation}
the transformation induced by the Type-III derangetropy functional modifies the spectral representation as follows:
\begin{equation}
	\varphi_{\nu}(t)=\varphi_0(t)-\frac{1}{2}\left[\varphi_0(t+2 \pi)+\varphi_0(t-2 \pi)\right].
\end{equation}

This transformation defines an operator $\mathcal{T}$ that acts recursively on the characteristic function, introducing frequency shifts at $\pm 2 \pi$ and enforcing a structured refinement process in Fourier space. The recursive application of $\mathcal{T}$ induces a diffusion-like behavior, where successive iterations lead to a spectral stabilization process governed by a structured decay. The following theorem formalizes this structured diffusion process, showing that the characteristic function of the sequence of transformed distributions satisfies a recurrence relation.

\begin{theorem}
	Let $X_n$ be a sequence of random variables whose characteristic function evolves under the transformation
	\begin{equation}
		\varphi_n(t)=\mathcal{T}\left[\varphi_{n-1}\right](t)=\varphi_{n-1}(t)-\frac{1}{2}\left[\varphi_{n-1}(t+2 \pi )+\varphi_{n-1}(t-2 \pi)\right],
	\end{equation}
	and define the re-normalized iterate	
	$$
	\widetilde{X}_n:=\sqrt{2 \pi^2 n}\left(X_n-m\right)
	$$	
	with corresponding $\widetilde{\varphi}_n(t)=\mathbb{E}\left[e^{i t \widetilde{X}_n}\right]$, where $m$ denotes the median of the distribution.
	Then, the scaled sequence $\widetilde{X}_n$ converges in distribution to a standard Gaussian. That is,
	$$
	\lim _{n \rightarrow \infty} \widetilde{\varphi}_n(t)=e^{-t^2 / 2}, \quad \text { for all } t \in \mathbb{R},
	$$	
	with exponential convergence rate $O\left(e^{-2 \pi^2 n}\right)$ in the variance collapse.
\end{theorem}

\begin{proof}
	We begin by analyzing the recursive evolution of the characteristic functions \( \varphi_n(t) \) under the transformation
	\begin{equation}
		\varphi_{n+1}(t) = \varphi_n(t) - \frac{1}{2} \left[ \varphi_n(t + h) + \varphi_n(t - h) \right], \quad \text{with } h = 2\pi.
	\end{equation}
	To approximate the right-hand side, we apply a second-order Taylor expansion of \( \varphi_n \) around the point \( t \), assuming sufficient smoothness, which yields
	\begin{equation}
		\varphi_n(t + h) + \varphi_n(t - h) = 2 \varphi_n(t) + h^2 \varphi_n''(t) + O(h^4).
	\end{equation}
	Substituting this expansion into the recurrence yields the following approximation
	\begin{equation}
		\varphi_{n+1}(t) = \varphi_n(t) - \pi^2 \varphi_n''(t) + O(h^4),
	\end{equation}
	which reveals that the recursive map approximately evolves according to a discrete second-order differential operator in Fourier space.
	
	Next, we consider the behavior of the variance of \( X_n \). Since \( \varphi_n(t) \) is the characteristic function of a random variable with median \( m \), and assuming finite second moments, a second-order expansion around the origin gives
	\begin{equation}
		\varphi_n(t) = 1 - \frac{1}{2} \sigma_n^2 t^2 + o(t^2), \quad \text{as } t \to 0,
	\end{equation}
	where \( \sigma_n^2 := \operatorname{Var}(X_n) \). Substituting this expansion into the evolution equation for \( \varphi_n \) leads to
	\begin{equation}
		\sigma_{n+1}^2 = \sigma_n^2 - 2\pi^2 + o(1), \quad \text{as } n \to \infty.
	\end{equation}
	This recurrence implies that the variance decreases monotonically with \( n \), and in particular, it decays linearly until reaching the scale \( \sigma_n^2 = O(1/n) \). Consequently, the sequence \( X_n \) collapses in \( L^2 \) toward the constant value \( m \), i.e.,
	\begin{equation}
		\sigma_n^2 \to 0 \quad \text{and} \quad X_n \xrightarrow{L^2} m.
	\end{equation}
	
	\noindent We now analyze the asymptotic behavior of the re-scaled random variables \( \widetilde{X}_n := \sqrt{2\pi^2 n}(X_n - m) \). Let \( \widetilde{\varphi}_n(t) := \mathbb{E}[e^{it \widetilde{X}_n}] \) denote the characteristic function of \( \widetilde{X}_n \), and define the auxiliary function
	\begin{equation}
		\psi_n(t) := \widetilde{\varphi}_n\left( \frac{t}{\sqrt{2\pi^2 n}} \right).
	\end{equation}
	This scaling is chosen so that \( \psi_n(t) \) captures the low-frequency behavior of \( \widetilde{\varphi}_n \) near the origin, where Gaussian approximations are valid.
	
	Expanding \( \widetilde{\varphi}_n \) to second order, and using the earlier estimate \( \operatorname{Var}(X_n) \sim \frac{1}{2\pi^2 n} \), we obtain the following recurrence
	\begin{equation}
		\psi_{n+1}(t) = \psi_n(t) \left(1 - \frac{t^2}{2n} + o\left(\frac{1}{n}\right) \right).
	\end{equation}
	Taking logarithms and summing yields
	\begin{equation}
		\log \psi_n(t) = \sum_{k=1}^n \log \left(1 - \frac{t^2}{2k} + o\left(\frac{1}{k} \right) \right) = -\frac{t^2}{2} \log n + o(\log n).
	\end{equation}
	Finally, by exponentiating both sides, we conclude that
	\begin{equation}
		\psi_n(t) = e^{ -\frac{1}{2}t^2 + o(1) } \longrightarrow e^{-t^2/2}, \quad \text{as } n \to \infty.
	\end{equation}
	Therefore, the characteristic functions \( \widetilde{\varphi}_n(t) \) converge pointwise to the standard Gaussian characteristic function. By L\'evy?s continuity theorem, this establishes
	\begin{equation}
		\widetilde{X}_n \xrightarrow{d} \mathcal{N}(0,1),
	\end{equation}
	as \( n \to \infty \), which completes the proof.
\end{proof}

The Type-III derangetropy transformation defines a structured probability refinement process that induces an diffusion mechanism in Fourier space. The evolution of the characteristic function follows a discrete approximation to the heat equation
\begin{equation}
	\frac{\partial \varphi_n}{\partial n}=-\pi^2 \frac{\partial^2 \varphi_n}{\partial t^2},
\end{equation}
which describes how probability mass redistributes over successive iterations. Unlike uniform diffusion, which spreads probability evenly, Type-III derangetropy suppresses high-frequency oscillations while preserving structured probability redistribution, leading to a controlled refinement process. The sinusoidal modulation term $\sin ^2(\pi F(x))$ plays a central role in shaping this evolution. Rather than introducing oscillatory probability mass adjustments, this term guides structured refinement by selectively suppressing high-frequency variations, ensuring that probability mass does not dissipate arbitrarily.

The iterative application of Type-III derangetropy ultimately converges to a Gaussian distribution. This convergence is a direct consequence of the underlying diffusion mechanism in Fourier space, where the transformation systematically smooths the probability density while preserving structured patterns. The final equilibrium state corresponds to a probability density with maximal entropy under the imposed constraints, reinforcing the well-known result that repeated refinement of a distribution under structured diffusion leads to a Gaussian limit.

\section{Conclusion and Future Work}
\label{Sec:Conclusion}

This work introduced a generalized framework for derangetropy functionals, information-theoretic operators designed to model the redistribution of probability mass in systems characterized by cyclical modulation, feedback dynamics, and structural fluctuation. By developing and analyzing three distinct classes of derangetropy transformations, Type-I, Type-II, and Type-III, we have demonstrated how entropy modulation and phase sensitivity can be systematically encoded into functional mappings on probability densities. Each class of transformation corresponds to a qualitatively distinct mode of informational dynamics: entropy-attenuating sharpening, entropy-amplifying dispersion, and oscillatory redistribution.

Through formal derivations, we established that these functionals are governed by nonlinear differential equations whose behavior reflects the underlying informational geometry. In particular, we proved that derangetropy, when applied recursively, induces a diffusion process governed by the heat equation in the spectral domain, ultimately converging to a Gaussian characteristic function. This convergence theorem offers a unifying analytical perspective across all three classes and furnishes a tractable means to examine long-term behavior under cyclic modulation.

The present formulation opens several avenues for further exploration. One natural extension involves developing multivariate analogues of derangetropy functionals. While the current framework operates on univariate distributions, many real-world systems, including multi-region neural circuits and high-dimensional ecological models, exhibit structured dependencies that require joint distributional analysis. Extending the derangetropy framework to accommodate tensor-valued or matrix-variate distributions may allow for modeling higher-order interactions and spatiotemporal coupling.

Another direction involves embedding derangetropy into learning architectures. For instance, entropy-modulated functionals can serve as non-parametric regularizers or adaptive filters in recurrent neural networks, where they may offer improvements in stability, sparsity, or interpretability. Likewise, incorporating derangetropy transformations into variational inference or generative modeling frameworks could enable new forms of distributional control grounded in cyclic or feedback-sensitive priors.

Finally, the connection between recursive derangetropy transformations and dynamical systems theory invites a deeper study of the informational stability and bifurcation behavior induced by repeated application. In particular, analyzing the fixed points, limit cycles, or attractors of derangetropy-driven flows may yield insight into the conditions under which information concentrates, disperses, or stabilizes under repeated structural modulation.

Collectively, the generalized derangetropy framework presented here offers a novel and analytically tractable approach to modeling structured information flow. By situating information redistribution within a family of entropy-sensitive and self-referential functionals, this work lays the foundation for future developments at the interface of information theory, dynamical systems, and probabilistic modeling.

\appendix
\section*{Appendix A}\label{app}
\noindent
Let us compute
$\int_{0}^{1} \dfrac{\sin (\pi z)}{z^z(1-z)^{1-z}} dz  $ by defining
\begin{equation}
	\begin{aligned}
		S
		\, := \, &  \, \int_{0}^{1} \dfrac{e^{i\pi z}}{z^z \, (1-z)^{1-z}}  \, dz   \\
		\, = \, &  \, \int_{0}^{1} \dfrac{1}{1-z} \, \dfrac{(1-z)^z}{z^z} \, e^{i\pi z} \, dz   \\
		\, = \, &  \,  \int_{0}^{1} \left(\dfrac{1}{1-z}\right) \, e^{z( i\pi + \log (1-z) - \log(z) )} \, dz \, .   \\
	\end{aligned}
\end{equation}
By substituting $t=\log (1-z) - \log z$, we obtain $$ z=\dfrac{1}{1+e^t} \, , $$ which in turn yields
\vspace{-6pt}
\begin{equation}
	\begin{aligned}
		S
		\, = \, &  \, \int_{-\infty+i\pi}^{+\infty+i\pi} \left(\dfrac{1}{1-e^t}\right) \, e^{\dfrac{t}{1-e^t}} \, dt \, .  \\
	\end{aligned}
\end{equation}
Now, consider the following function
\vspace{-3pt}
\begin{equation}
	g(u) =\left(\dfrac{1}{1-e^u}\right) \, e^{\dfrac{u}{1-e^u}} \, .
\end{equation}
This function is meromorphic on the strip $$ D = \{ u\in \mathbb{C}: -\pi \leq \Im(u) \leq \pi \} $$ and its only point is located at $u=0$ having
\begin{equation}
	Res(g,0) = -\frac{1}{e} \, .
\end{equation}
Next, consider the rectangular clockwise path $\alpha$, as~depicted in Figure~\ref{fig:Contour}, such that $\alpha = \alpha_1 \oplus \alpha_2 \oplus \alpha_3 \oplus \alpha_4$, implying that
\vspace{-1pt}
\begin{equation}
	\oint_{\alpha}g(u) du
	\, = \,   \, -2\pi i\, Res(g,0) \, = \,   \,\, 2 i \, \frac{\pi}{e} \, . 
\end{equation}	
\noindent It is noted that for any sequence $\{u_n\}\subset D$, we have $|u_n|\to \infty$. In~turn, this leads $g(u_n)$ to approach infinity, implying that $\int_{\alpha_2} g(u)du$ and $\int_{\alpha_4} g(u)du$ both become zero as $R\to \infty$. By~further noting that $\int_{\alpha_1} g(u)du = \int_{\alpha_3} g(u)du$, one obtains 
\begin{equation}
	2S = 2\lim\limits_{R \to \infty} \oint_{\alpha_1} g(u) du = 2i\, \dfrac{\pi}{e}
\end{equation}
Thus,
\begin{equation}
	\int_{0}^{1} \dfrac{\sin (\pi z)}{z^z(1-z)^{1-z}} dz \, = \, \Im(S) \, =  \, \frac{\pi}{e}.
\end{equation}

\vspace{-12pt}
\begin{figure}[H]
	
	\includegraphics[scale=0.15]{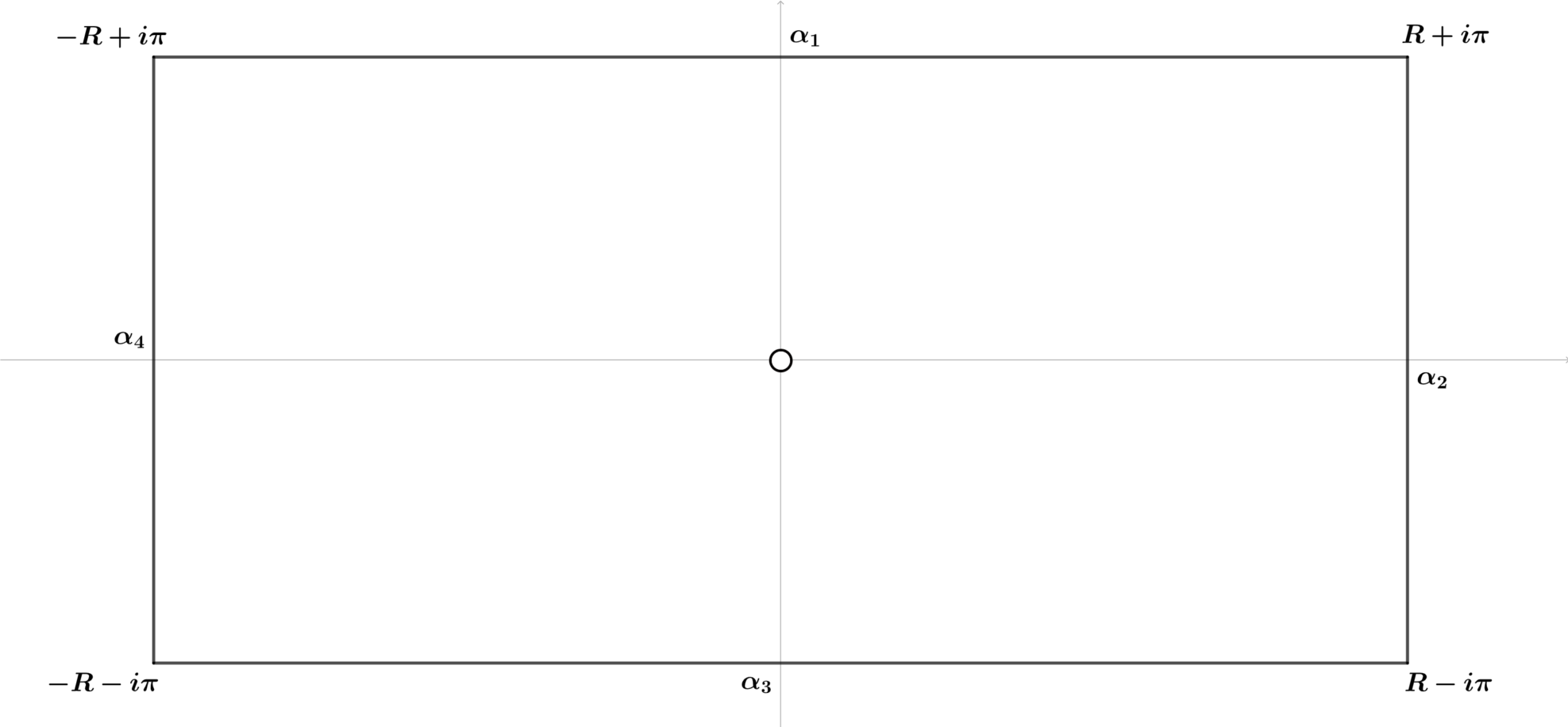}
	\caption{Schematic of the rectangular clockwise integration~path.}
	\label{fig:Contour}
\end{figure}

\section*{Appendix B}\label{char}
\noindent

\noindent	By definition, the characteristic function of $\nu_f(x)$ is
\begin{equation}
	\varphi_{\nu_f}(t)=\int_{-\infty}^{\infty} e^{i t x} \nu_f(x) d x.
\end{equation}
Substituting the definition of the Type-III transformation, we get
\begin{equation}
	\varphi_{\nu_f}(t)=2 \int_{-\infty}^{\infty} e^{i t x} \sin ^2(\pi F(x)) f(x) d x.
\end{equation}	
Thereafter, using the trigonometric identity
\begin{equation}
	\sin ^2(\pi F(x))=\frac{1-\cos (2 \pi F(x))}{2},
\end{equation}	
we rewrite the expression as follows:
\begin{equation}
	\varphi_{\nu_f}(t)=\int_{-\infty}^{\infty} e^{i t x} f(x) d x-\frac{1}{2} \int_{-\infty}^{\infty} e^{i t x} \cos (2 \pi F(x)) f(x) d x.
\end{equation}
The first integral is simply the characteristic function of the original density, where
\begin{equation}
	\varphi_0(t)=\int_{-\infty}^{\infty} e^{i t x} f(x) d x.
\end{equation}
Moreover, using the Euler's identity, we obtain
\begin{equation}
	\cos (2 \pi F(x))=\frac{e^{i 2 \pi F(x)}+e^{-i 2 \pi F(x)}}{2},
\end{equation}	
which leads to	
\begin{equation}
	\int_{-\infty}^{\infty} e^{i t x} \cos (2 \pi F(x)) f(x) d x=\frac{1}{2}\left[\int_{-\infty}^{\infty} e^{i t x} e^{i 2 \pi F(x)} f(x) d x+\int_{-\infty}^{\infty} e^{i t x} e^{-i 2 \pi F(x)} f(x) d x\right].
\end{equation}
Finally, defining the modulated characteristic functions by	
\begin{equation}
	\varphi_F^{+}(t) =\int_{-\infty}^{\infty} e^{i t x} e^{i 2 \pi F(x)} f(x) d x, \\
\end{equation}	
and
\begin{equation}
	\varphi_F^{-}(t) =\int_{-\infty}^{\infty} e^{i t x} e^{-i 2 \pi F(x)} f(x) d x,
\end{equation}
we obtain
\begin{equation}
	\varphi_{\nu_f}(t)=\varphi_0(t)-\frac{1}{2}\left[\varphi_F^{+}(t)+\varphi_F^{-}(t)\right].
\end{equation}

\newpage
\bibliography{references}

\end{document}